\documentclass[11pt]{llncs}

\usepackage{ifthen}

\newboolean{shortdraft}
\newboolean{qcircuit}

\setboolean{qcircuit}{true}     

\ifthenelse{\boolean{shortdraft}}{
\usepackage[paperheight=8in, paperwidth=6.5in, margin=0.2in]{geometry}
}{
\usepackage[margin=1in]{geometry}}

\usepackage{amsmath}
\usepackage{amsfonts}
\usepackage{amssymb}
\usepackage{url}
\usepackage{hyperref}

\ifthenelse{\boolean{qcircuit}}{
	\input{Qcircuit}
}{
\newcommand{\ket}[1]{\left | \, #1 \right \rangle}
\newcommand{\bra}[1]{\left \langle #1 \, \right |}
}

\ifthenelse{\boolean{shortdraft}}{
	\renewcommand{\cite}[1]{[??]}
}{
	\bibliographystyle{halphads}
}

\newcommand{\proj}[2]{\ket{#1}\!\!\bra{#2}}

\newcommand{\tr}[2][ ]{\text{Tr}_{#1}\!\left( #2 \right)}

\begin{document}

\title{Generalized self-testing and the security of the 6-state protocol}

\author{Matthew McKague \inst{1} and Michele Mosca\inst{1,2}}
\institute{Institute for Quantum Computing and Department of Combinatorics, University of Waterloo, Waterloo, Ontario N2L 3G1, Canada \and Perimeter Institute for Theoretical Physics, 31 Caroline Street
North, Waterloo, ON, N2L 2Y5, Canada}

\maketitle

\begin{abstract}
Self-tested quantum information processing provides a means for doing useful
information processing with untrusted quantum apparatus.
Previous work was limited to performing computations and protocols in real
Hilbert spaces, which is not a serious obstacle if one is only interested in
final measurement statistics being correct (for example, getting the correct
factors of a large number after running Shor's factoring algorithm). This
limitation was shown by McKague et al. to be fundamental, since there is no
way to experimentally distinguish any quantum experiment from a special
simulation using states and operators with only real coefficients.

In this paper, we show that one can still do a meaningful self-test of
quantum apparatus with complex amplitudes. In particular, 
we define a family of simulations of quantum experiments, based on complex conjugation, with two
interesting properties.  First, we are able to define a self-test which may
be passed only by states and operators that are equivalent to simulations
within the family.  This extends work of Mayers and Yao and Magniez et al.
in self-testing of quantum apparatus, and includes a complex measurement.  Second, any of the simulations in the family may be used to implement a secure 6-state QKD protocol, which was previously not known to be implementable in a self-tested framework.

\end{abstract}

\section{Introduction}

In \cite{Mayers:2004:Self-testing-qu}, \cite{Mayers:1998:Quantum-Cryptog}, Mayers and Yao introduced the concept of self-testing quantum apparatus with a test for EPR sources and a select set of measurements.  In a parallel development, van Dam et al. \cite{van-Dam:1999:Self-Testing-of} introduced the notion of self-testers for quantum circuits in the case where the dimension of the Hilbert space is known.  These results were then combined and improved upon by Magniez et al. in \cite{Magniez:2006:Self-testing-of}, who give a construction for self-testable circuits without knowledge of the dimension of the Hilbert space.

The Mayers-Yao test, and the test of Magniez et al., only allowed for the testing of states and operators that are equivalent to states and operators in a real Hilbert space.  McKague et al. \cite{McKague:2009:Simulating-Quan} showed that in such settings with untrusted apparatus, one cannot experimentally distinguish a quantum system with states and evolution involving complex amplitudes from a special simulation using only real amplitudes.  In addition to the implications for self-testing untrusted quantum apparatus, this also resolved an open question posed by Gisin \cite{Gisin:2007:Bell-inequaliti} related to the violation of Bell inequalities.  It is important to note that the real simulation does not preserve inner product relationships from the system it is simulating.  At first glance, this suggests that the well-known 6-state protocol \cite{Bennett:1984:Eavesdrop-detec}, \cite{Bruss:1998:Optimal-Eavesdr} might not be secure in a setting with untrusted apparatus, since the simulated versions of the six quantum states could be more distinguishable than the proofs of security assume, and an adversary could exploit this additional distinguishability and compromise security.

In fact, it is easy to describe such an insecure implementation of the
6-state protocol with untrusted apparatus, however even an implementation of
standard BB84 quantum key establishment with untrusted apparatus is insecure
if proper measures are not taken in order to rule out ``side-channel''
attacks. We show that with comparable precautions as those proposed by
Mayers and Yao the 6-state protocol remains secure.

This paper starts by describing a general family of simulations that will reproduce the same statistics as any given ``reference'' experiment, and are thus experimentally indistinguishable from said experiment.  We show how the real Hilbert space simulation given in \cite{McKague:2009:Simulating-Quan} is equivalent to a special case of this family of simulations.  The fact that these simulations work is not very surprising: they are essentially mixtures of the reference experiment, or the complex conjugate of the reference experiment. Thus, we have a more general collection of experiments that are experimentally indistinguishable in a self-testing framework. What is particularly remarkable is that we are able to describe, in section~\ref{sec:extmayersyaoref}, a family of self-tests which can only be passed by simulations from the general family we describe (up to equivalence, as defined below).  This is summarized in Theorem~\ref{theorem:extmayersyao}.

The self-tests allow us to put a physical experiment in a general collection of experiments, and we are then able to show that the 6-state protocol is secure for all the experiments within this collection.  This shows that it is possible to define a secure 6-state protocol within the self-testing framework.

In section~\ref{sec:mayersyao}, we prepare for the proof of Theorem~\ref{theorem:extmayersyao}, by discussing the Mayers-Yao self-tested source result given in Theorem~\ref{theorem:mayers-yao}, a new proof of which is given in appendix~\ref{sec:mayers-yaoproof}.  This new proof is shorter and simpler, and more easily extended to prove our more general result.

Then, in section~\ref{sec:extmayersyaoref}, we describe a new self-test for an EPR source and local measurement apparatus that will uniquely characterize the general equivalence class associated with this quantum state and measurement operators.

In section~\ref{sec:selftestcrypto}, we discuss the cryptographic implications, and why a properly
self-tested 6-state protocol is still secure.

Lastly, in section~\ref{sec:discussion}, we discuss some open problems and future directions,
including the robustness of the generalized self-tests.

\section{Simulations}\label{sec:simulations}
In this section we extend the work of McKague et al. in \cite{McKague:2009:Simulating-Quan}.  There the authors gave a construction that allowed the outcomes of any experiment (the \emph{reference} experiment) to be duplicated (\emph{simulated}) by another experiment (the \emph{simulation}) which is described entirely using real numbers.  That is to say, all the states, measurement operators, unitaries, Kraus operators, and Hamiltonians are given as vectors and matrices over the real numbers.  Of particular interest here is the fact that the simulation is, in general, not equivalent to the reference experiment according to definition~\ref{def:equivalence} below.

In this section we give a construction for a wider family of simulations.  The different simulations in the family are, in general, not equivalent to either the reference experiment nor each other.  We will be most interested in states and measurements, but, as with the real simulation in \cite{McKague:2009:Simulating-Quan}, it is also possible to simulate discrete and continuous time evolution.

The simulations rely on the simple observation that transforming an experiment by complex conjugation does not alter the statistics it generates.  We could also take a classical mixture of the reference experiment and its complex conjugation, flipping a coin (or controlling on a qubit) beforehand to decide which one to perform.  In the remainder of this section we fill in some details about the simulations defined by these mixtures.

\subsection{States and measurements}

Consider a reference state\footnote{We may consider mixed states as well, but it is not necessary for our discussion.} $\ket{\psi}$ measured according to a reference POVM $\{ P_{k}\}$.  We may duplicate the statistics of this experiment using the complex conjugate state $\ket{\psi^{*}}$ and POVM $\{P_{k}^{*}\}$.  In addition, we could do some combination of the two; we may add an additional qubit register which records which of the two experiments to perform:  $\ket{0}$ for the reference experiment, and $\ket{1}$ for the complex conjugate.  This qubit may be in any state, and not necessarily pure.  We then arrive at a new state
\begin{equation}\label{eq:simstate}
\rho^{\prime} = a \proj{0}{0} \otimes \proj{\psi}{\psi} + (1-a) \proj{1}{1} \otimes \proj{\psi^{*}}{\psi^{*}} + c \proj{0}{1}\otimes \proj{\psi}{\psi^{*}} + c^{*} \proj{1}{0} \otimes \proj{\psi^{*}}{\psi}
\end{equation}
with $a \geq 0$ and $|c| \leq \sqrt{a(1-a)}$.  The important feature is that when we project onto $\proj{0}{0}$ or $\proj{1}{1}$ we get either $\ket{\psi}$ or $\ket{\psi^{*}}$, respectively.  For the measurement, we form the POVM
 \begin{equation}\label{eq:cont_sim_povm}
\{ \proj{0}{0} \otimes P_{k} + \proj{1}{1} \otimes P^{*}_{k}\}.
\end{equation}
This POVM measurement is equivalent to measuring the added qubit, collapsing the state into either $\ket{\psi}$ or $\ket{\psi^{*}}$ and then measuring either $\{P_{k}\}$ or $\{P_{k}^{*}\}$ as appropriate; thus the statistics of the experiment are preserved.

Different simulations are arrived at by choosing different values of $a$ and $c$.  If $a=1$ and $c=0$ then we obtain the reference experiment.  For $a=0$ and $c=0$ we obtain the complex conjugate.  Another interesting case is when $a=c=\frac{1}{2}$, in which case we obtain (up to a local change of bases) the real simulation of \cite{McKague:2009:Simulating-Quan} as shown in appendix~\ref{sec:real_sim_in_cont_sim}.

\subsection{Operators}
Although it will not be necessary for our discussion, it is possible to simulate a reference experiment which includes evolution, according to a unitary, completely-positive map, or Hamiltonian.  The details are discussed in appendix~\ref{sec:evoinsimulations}.

\subsection{Non-local computations}
For multi-party experiments, such as the Mayers-Yao test, we would need the simulation to be performed in a local fashion with the measurements operating on local systems only.  As defined above this not the case, but it is easy to modify the operators to make it so.  We simply add an extra qubit for each party and record in each qubit whether to perform the reference experiment or the complex conjugate.  We arrive at states analogous to that in equation~\ref{eq:simstate}, but with $\ket{0}$ and $\ket{1}$ replaced with logical states $\ket{\overline{0}} = \ket{00\dots0}$, $\ket{\overline{1}}=\ket{11\dots1}$ defined on the extra qubits held by the various parties.  Finally, each party conditions their operations on their local copy of the qubit, applying either the reference operation or the complex conjugate.

\section{Mayers-Yao self-test}\label{sec:mayersyao}
The goal of the Mayers-Yao test is to compare two experiments.  The first experiment is the \emph{reference} experiment, which is the experiment we wish to implement.  It is a blueprint, or gold standard, against which we compare the other experiment, the physical experiment, which is the experiment that is actually performed.  Within the physical experiment we consider the entire physical apparatus, including the environment, so that we obtain a pure state on a Hilbert space of unknown dimension (however, we will limit ourselves to finite dimensions.)  The reference and physical experiments consist of reference and physical states, operations, and measurements.  The two experiments are compared through the statistics that they generate.

\subsection{Equivalence}

The proof considers a particular reference experiment, as described in section~\ref{sec:reftest}.  This experiment is defined on a pair of qubits, so we will limit our discussion to such systems.  As well, we consider only pure states - the physical system is unlimited (but finite) in size, so we may include the environment to obtain a pure state.  The conclusion of Mayers and Yao is that if the statistics of a physical experiment agree with that of the reference experiment, then the physical experiment is equivalent to the reference experiment, under a particular notion of equivalence.

When defining a notion of equivalence in this setting we must first consider how me might change the reference experiment in a way that preserves the statistics of the outcomes.  Any such change is invisible from the perspective of the statistics and hence we cannot rule them out.  Here is a list of such changes:

\begin{enumerate}
	\item Local changes of basis
	\item Adding ancillae to physical systems, prepared in any joint state (the measurement does not act on them)
	\item Changing the action of the observables outside the support of the state
	\item Locally embedding the state and operators in a larger (or smaller) Hilbert space.
\end{enumerate}

In order to accommodate these various changes we define equivalence as follows.

\begin{definition}\label{def:equivalence}
 A {\it reference experiment} is described by a $n$-partite state $\ket{\psi}$ on Hilbert space $\mathcal{X} = \mathcal{X}_{1} \otimes \dots \mathcal{X}_{n}$ and local measurement observables $M_{m}$ for various $m$.  Further, consider a physical experiment described by a $n$-partite state $\ket{\psi^{\prime}}$ on Hilbert space $\mathcal{Y} = \mathcal{Y}_{1}\otimes \dots \otimes \mathcal{Y}_{n}$ and local measurement observables $M^{\prime}_{m}$ for various $m$.  We say that the physical experiment is \emph{equivalent}\footnote{Note that this is not an equivalence relation since it is not symmetric.} to the reference experiment (and the physical state and measurement observables are equivalent to the reference state and measurement observables) if there exists a local isometry
\begin{equation}
\Phi = \Phi_{1} \otimes \dots \Phi_{n}, \, \, \, \,
\Phi_{j} : \mathcal{Y}_{j} \mapsto \mathcal{Y}_{j} \otimes \mathcal{X}_{j}
\end{equation}
 such that
\begin{eqnarray}
\Phi(\ket{\psi^{\prime}}) & = & \ket{junk}_{\mathcal{Y}} \otimes \ket{\psi}_{\mathcal{X}}\\
\Phi(M^{\prime}_{m}\ket{\psi^{\prime}}) & = & \ket{junk}_{\mathcal{Y}} \otimes M_{m}\ket{\psi}_{\mathcal{X}}.
\end{eqnarray}

\end{definition}

The isometry $\Phi$ may be constructed by attaching ancillae in some product state $\ket{00\dots 0}_{\mathcal{X}}$ and applying local unitaries to the subsystems.  Note that if we make any finite number of changes from the list above then we may construct a suitable local isometry and show that the experiment is equivalent to the reference experiment.  Also, any experiment that is equivalent to the reference experiment may be constructed by applying changes from the list above:  one simply attaches ancillae in the state $\ket{junk}$ and performs a suitable change of basis.  The content of the main theorem is that, for a carefully chosen experiment, these are the \emph{only} changes that preserve the statistics.

\begin{theorem}[Mayers and Yao]\label{theorem:mayers-yao}
Suppose a physical experiment reproduces the statistics of the reference experiment described in section~\ref{sec:reftest}.  Then the physical experiment is equivalent to the reference experiment.
\end{theorem}

A simplified proof for the Mayers-Yao self-test is give in appendix~\ref{sec:mayers-yaoproof}.

\subsection{Mayers-Yao self-test reference experiment}\label{sec:reftest}
A general schematic for the Mayers-Yao reference experiment is shown in figure~\ref{eprtest}.  A bipartite state $\ket{\psi}$ is distributed to a pair of measurement devices.  The two measurement devices take classical inputs $a$ and $b$, which each take one of three values.  The devices then output classical bits, $x$ and $y$.

 \begin{figure}
\[ \Qcircuit {
 & & a \\
 & & \cwx[1] \\
  & & \meter & \cw & x \\
\lstick{\ket{\psi}}   \ar@{-}[ur] \ar@{-}[dr]\\
  & & \meter & \cw & y \\
  & & \cwx[-1] \\
  & & b \\
 }
\]
\caption{Mayers-Yao self-test circuit}
\label{eprtest}
\end{figure}
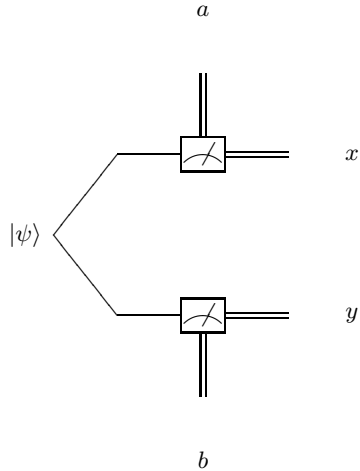

The reference state is an EPR pair $\ket{\phi_{+}} = \frac{1}{\sqrt{2}}\left(\ket{00} + \ket{11} \right)$ and the reference measurement observables are $X, Z, \frac{X + Z}{\sqrt{2}}$ for each side of the EPR pair.  For brevity we label $\frac{X + Z}{\sqrt{2}} = D$.  For the untrusted physical devices this equality is not given, so there the separate label $D$ is required.

\section{Extending the Mayers and Yao self-test}\label{sec:extmayersyao}
The original Mayers and Yao EPR test utilized only a small set of measurements.  Conspicuously missing is anything with complex coefficients.  An important consequence of this is that the circuit test developed by Magniez et al. \cite{Magniez:2006:Self-testing-of} is not able to test gates with complex coefficients; only gates with real coefficients can be tested.  More specifically, real measurements reveal no information about the imaginary component of a density matrix.

In fact the Mayers-Yao self-test cannot be directly extended to include any measurements with complex coefficients.  This is a result of the notion of equivalence used.  Suppose that we wish to include the $Y$ measurement in the set of reference measurements.  The devices could instead implement $-Y$, the complex conjugate.  So long as all complex measurements were complex conjugated it would be impossible to detect this change.  Although this does not present an immediate problem - such a transformation is internally consistent and produces the correct outcome statistics - we cannot transform such a circuit back into the reference circuit using unitary transformations.

If this were the whole story we could simply require that the physical circuit be transformable into either the reference circuit or its complex conjugate.  However, the real simulation, and now the general family of simulations, are also indistinguishable from the reference circuit and not unitarily transformable into the reference circuit.

We have one encouraging fact:  all of the known simulations are equivalent to a simulation from the general family of simulations.  We now prove that we can extend the Mayers-Yao test such that these are the only simulations.  Hence we may extend our notion of equivalence to include these simulations and obtain a new self-testing theorem.

\begin{theorem}\label{theorem:extmayersyao}
Suppose a physical experiment duplicates the statistics generated by the reference experiment described in section~\ref{sec:extmayersyaoref}.  Then the physical experiment is equivalent to one of the simulations of the reference experiment described in section~\ref{sec:simulations}.
\end{theorem}

\subsection{Extended Mayers-Yao self-test reference experiment}\label{sec:extmayersyaoref}
The extended Mayers-Yao test will consist of three regular Mayers-Yao tests, performed together.  Alice and Bob will perform the Mayers-Yao test with measurement settings (labelled with subscript $A$ when used by Alice, and subscript $B$ when used by Bob):
\begin{enumerate}
                \item $X$, $Z$, and $D$
                \item $X$, $Y$, and $E$
                \item $Y$, $Z$, and $F$.
\end{enumerate}

In the reference experiment the measurement settings $X$, $Y$ and $Z$ are realized by the Pauli operators, with $Y_B=-Y$ and otherwise
$X_A = X_B = X$, $Y_A = Y$, $Z_A  = Z_B = Z$.  The other settings are realized by $D_A = \frac{X+Z}{\sqrt{2}}$, $E_A = \frac{X+Y}{\sqrt{2}}$, $F_A=\frac{Y+Z}{\sqrt{2}}$ on Alice's side and $D_B = \frac{X+Z}{\sqrt{2}}$, $E_B = \frac{X-Y}{\sqrt{2}}$, $F_B=\frac{Z-Y}{\sqrt{2}}$ on Bob's side.  Bob's $Y_B$ measurements all carry the $-1$ phase since measuring the state $\ket{\phi_{+}}$ with the operator $Y\otimes Y$ produces $-1$ instead of $1$ as in the Mayers-Yao reference experiment.  The reference state is again $\ket{\phi_{+}}$.

\subsection{Proof of Theorem~\ref{theorem:extmayersyao}}
We start by assuming that the states are all pure as in the Mayers-Yao test.  Again we may incorporate the purification of a mixed state into either Alice or Bob's state by adding an ancilla.

First we apply the Mayers-Yao result with the measurements $X$, $Z$ and $D$.  We find a local isometry $\Phi$ as in definition~\ref{def:equivalence}.  We extend $\Phi$ by adding an extra qubit for each side initialized in the state $\ket{0}$. Then $\Phi$ takes the $X_{A}$, $Z_{A}$, $X_{B}$ and $Z_{B}$ measurements to $X_{Q_{A}} \otimes I_{R_{A}} \otimes I_{S_{A}}$, $Z_{Q_{B}} \otimes I_{R_{A}} \otimes I_{S_{A}}$, $X_{Q_{B}} \otimes I_{R_{B}} \otimes I_{S_{B}}$ and $Z_{Q_{B}} \otimes I_{R_{B}} \otimes I_{S_{B}}$ where $R_{A}$ and $R_{B}$ are the added qubit registers and $S_{A}$ and $S_{B}$ are the junk registers.  Meanwhile the state has the form $\ket{\phi_{+}}_{Q_{A}Q_{B}} \otimes \ket{00}_{R_{A}R_{B}} \otimes \ket{junk}_{S_{A} S_{B}}$.

We now consider the remaining measurements.  The reference experiments for these measurements can be transformed using local unitaries into the usual Mayers-Yao reference experiments.  Thus we may apply the result.  However, we stop short of using the full result.  Within the proof of Theorem~\ref{theorem:mayers-yao} we achieve the following result.

\begin{lemma}
Suppose a physical experiment reproduces the statistics of the Mayers-Yao reference experiment described in section~\ref{sec:reftest}.  Then the physical measurements $X_{A}$ and $Z_{A}$ anti-commute on the support of the physical state, as do $X_{B}$ and $Z_{B}$.
\end{lemma}
This is shown in section~\ref{sec:anticommute}.  When we apply this result to the remaining measurements in the extended test, we find that $X_{A}$ and $Y_{A}$ anti-commute on the support of the state, as do $X_{B}$ and $Y_{B}$, $Z_{A}$ and $Y_{A}$ and $Z_{B}$ and $Y_{B}$.  For the remaining discussion we will limit ourselves to the support of the state.

Consider the $A$ side measurements first.  We may express $Y_{A}$ as
 \[
Y_{A} = \sum_{P,E} y_{P,E} P_{Q_{A}} \otimes I_{R_{A}} \otimes E_{ S_{A}}
\]
where the $P$ ranges over the Pauli operators and the $E$ ranges over a basis for the Hermitian operators on $S_{A}$.

Since $Y_{A}$ anti-commutes with $X_{Q_{A}} \otimes I_{R_{A} S_{A}}$ the coefficients of all the terms with $P=X$ must be 0.  Indeed, since $-Y_{A} = (X_{Q_{A}} \otimes I_{R_{A}S_{A}}) Y_{A} (X_{Q_{A}} \otimes I_{R_{A}S_{A}})$ we have
 \[
-\sum_{P,E} y_{P,E} P_{Q_{A}} \otimes I_{R_{A}} \otimes E_{S_{A}}
 = \sum_{P\in\{I,X\}, E}y_{P,E} P_{Q_{A}} \otimes I_{R_{A}} \otimes E_{S_{A}} -  \sum_{P \in\{Y,Z\}, E} y_{P,E} P_{Q_{A}} \otimes I_{R_{A}} \otimes E_{S_{A}}
\]
where on the right hand side we have separated out the terms that commute with $X_{Q_{A}} \otimes I_{R_{A}S_{A}}$ and those that anti-commute.  We see that we must have $y_{X,E} = -y_{X,E} = 0$ and $y_{I,E} = -y_{I,E} = 0$ for all $E$.

Applying similar reasoning and the test with $Y$ and $Z$ we find that $y_{Z,E} = 0$ for all $E$.  Thus $Y_{A} = Y_{Q_{A}} \otimes I_{R_{A}} \otimes M_{S_{A}}$ for some Hermitian and unitary $M_{S_{A}}$.  Next we compose $\Phi$ with a ``phase kickback'' circuit consisting of a Hadamard gate on the $R_{A}$ register, followed by a controlled $M_{S_{A}}$, controlled on the $R_{A}$ register, and a final Hadamard gate on the $R_{A}$ register.  This results in a new isometry (we will still call it $\Phi$) such that
\begin{equation}
\Phi(Y_{Q_{A}} \otimes M_{S_{A}} \ket{\psi}) = Y_{Q_{A}} \otimes Z_{R_{A}} \ket{\phi_{+}}_{Q}  \ket{junk}_{RS}.
\end{equation}
This is essentially the well known translation of a two outcome measurement into a qubit measurement.  Also, since the addition of the phase kickback did not operate on the junk register the $X$ and $Z$ measurements are not affected.

The above process can be repeated for Bob's side, with analogous conclusions.  In order to be consistent with the reference experiment, we may construct our isomorphism so that 
\begin{equation}
\Phi(Y_{Q_{B}} \otimes M_{S_{B}} \ket{\psi}) = Y_{Q_{B}} \otimes Z_{R_{B}} \ket{\phi_{+}}_{Q}  \ket{junk}_{RS}.
\end{equation}
We have thus shown that the measurements are as in the general simulation.

We now turn our attention to the state.  From the Mayers-Yao test on $X$ and $Z$ we know that the state on $Q_{A} \otimes Q_{B}$ (after applying $\Phi$) is $\ket{\phi_{+}}$.  We next consider the state on the remaining registers, $\ket{junk}_{RS}$.  We may express this in the singular value (Schmidt) decomposition, split between $R_{AB}$ and $S_{AB}$:
\begin{equation}
\ket{\theta} = \sum_{j} \lambda_{j} \ket{j}_{R_{AB}} \ket{j}_{S_{AB}}
\end{equation}
with $\lambda_{j} > 0$.
Since the $Y$ measurement setting gives correlated results (recall we introduced a -1 factor on the $B$ side measurement observable) and the form of $Y_{A}$ and $Y_{B}$, the states $\ket{j}_{R_{AB}}$ must all be $+1$ eigenvectors of $Z_{R_{A}}\otimes Z_{R_{B}}$.  If this were not the case then a $-1$ phase would be introduced and the measurement results would be incorrect at least some of the time.  Thus the only possible states for $\ket{j}_{R_{AB}}$ are superpositions of $\ket{00}$ and $\ket{11}$.  We do some relabelling and arrive at
 \begin{equation}
\ket{\psi} = \ket{\phi_{+}}_{Q_{AB}} \otimes \left(\alpha \ket{00}_{R_{AB}} \ket{\theta_{00}}_{S_{AB}} + \beta \ket{11}_{R_{AB}} \ket{\theta_{11}}_{S_{AB}}\right)
\end{equation}
with $\ket{\theta_{00}}$ and $\ket{\theta_{11}}$ not necessarily orthogonal.  Note that tracing out the $S_{AB}$ ancillae results in a state exactly as described by the multi-party simulation in section~\ref{sec:simulations}.  Thus we have demonstrated that the physical experiment is equivalent to one of the general simulations of the reference experiment, and completed the proof of Theorem~\ref{theorem:extmayersyao}.

\section{Cryptographic setting}\label{sec:selftestcrypto}

Suppose that two or more parties are engaged in a cryptographic protocol using self-tested apparatus.  The extended Mayers-Yao test above allows them to determine that the devices are implementing a simulation from the family of simulations described in section~\ref{sec:simulations}.  Suppose further that the adversary, Eve, knows how the devices are implemented and controls the preparation of the state.  The honest parties only perform operations as specified for the simulation.  Eve, on the other hand, is free to interact with the extra qubits in the simulation in any way she likes.  Does this give any advantage to Eve?

Eve can potentially perform many operations, including entangling a qubit of her own with the extra simulation qubits allowing her to perform simulation operations.  She may also interact in complex ways with the extra simulation qubits along with the original register.  Despite this, we are able to prove that Eve can gain no advantage for some protocols.

We explore a restricted class of protocols that are especially easy to analyse.  These are protocols where the only operation that an honest party will do is a Pauli measurement.  This class includes the six-state quantum key distribution protocol (implemented in as an entanglement based protocol) \cite{Bennett:1984:Eavesdrop-detec}, \cite{Bruss:1998:Optimal-Eavesdr}.  We will demonstrate that these protocols do not leak any more information when implemented using one of the simulations.

The proof is a series of security reductions to protocols in which each reduction only increases Eve's power.  We will show that the final protocol in the reduction is just as secure as the reference protocol (without the simulation applied), hence the simulation protocol is also just as secure as the reference protocol.

For the first reduction we suppose that the participants in the protocol measure their simulation qubit in the $Z$ eigenbasis after the protocol is completed, and transmit the result to Eve.  This does not interfere with the intended protocol and only increases Eve's information.  Since the $Z$ measurement commutes with all simulation operations, the participants could just as well have performed the measurement before the protocol began.  If Eve is the one who prepares the initial state for the simulation (in other cases Eve has strictly less power) then Eve could also perform this measurement herself.  This measurement would collapse the state to an eigenvector of the $Z$ measurements and Eve's strategy would be a mixture of different strategies with the states each an eigenvector of the $Z$ measurements.

Let us examine the result of Eve choosing one of these eigenvector states.  Each of the parties will receive their extra qubit prepared in a $Z$ eigenvector.  The effect of this on their operations is either to perform the protocol's original operation (in the case of a $\ket{0}$) or the complex conjugate (in the case of a $\ket{1}$.)  For Pauli measurements, only the $Y$ measurement is affected: the output bit is flipped in the case of the complex conjugate.

If every party receives the same eigenvector in their extra qubit, then the protocol reduces to either the original or the complex conjugate.  In either case the security is identical to the original protocol.  If the extra qubits are not in the same eigenvector then some $Y$ measurements outcomes will be flipped and some will not.  This does not affect Eve's information since she controls which outcomes are flipped and can undo the flips in her reckoning of the final classical information.  Note that the bit flips may introduce errors into the protocol.  If the protocol does not explicitly check for such errors (as does the 6-state protocol) information will still not be leaked to Eve, however a test for these errors may be required to make sure the protocol functions correctly.  The final protocol, and hence the simulation, is thus as secure as the original protocol.

\section{Conclusions and Future work}\label{sec:discussion}

\subsection{Conclusions}
Theorem~\ref{theorem:extmayersyao}, along with the security result of section~\ref{sec:selftestcrypto}, allows us to analyze the case of the 6-state QKD protocol in the self-tested framework.  In particular we may define a self-testing version of the 6-state protocol in which the extended Mayers-Yao test is incorporated along with the usual 6-state protocol.  Given a robust version of the test (see section~\ref{sec:futurework}) we may first estimate the state and measurement observables, then apply a security proof for the 6-state protocol in order to derive a secure key rate.

Although a self-tested 6-state protocol is currently not practical, nor likely to become so, the result is interesting from a theoretical perspective within the self-tested framework.  Previous results were limited to real Hilbert spaces, one could apply the real simulation explicitly within the reference experiment and then proceed with the self-test.  This works fine for circuits, where only the correct outcome is important, however the 6-state protocol introduces other concerns, namely the possibility of information leaking to an adversary.  The current work thus illustrates how a self-test for complex operations provides additional benefit over the previous self-tests.

\subsection{Future work}\label{sec:futurework}
Note that we have not described a physically realizable test in section~\ref{sec:extmayersyao}.  The proof requires that the expected value of the observables match the reference exactly.  This cannot be established physically without some kind of repeatability assumptions and an infinite number of trials.   The original test by Mayers and Yao was shown to be robust in \cite{Magniez:2006:Self-testing-of}, establishing a polynomial relationship between the precision of the statistics and the closeness to an EPR state. We are currently studying the robustness of these new tests.  This is an important line of future research.  A related task is to extend the results to continuous variable systems.

Another interesting line of research is to follow the same path as Magniez et al. to obtain a self-testing circuit for arbitrary circuits, now allowing complex gates.  The framework and proofs from \cite{Magniez:2006:Self-testing-of} offer a roadmap for such research, but there are some technical problems that arise along the way so a straightforward adaptation is not possible.  These are due to the larger Hilbert space created when adding the extra qubits to allow the simulations.

\subsubsection{Acknowledgements}

This work is supported by Canada's NSERC, QuantumWorks, Ontario Centres of Excellence, MITACS, CIFAR, CRC, ORF, the Government of Canada, and Ontario-MRI.

\bibliography{Global_Bibliography}

\appendix

\section{Evolution in simulations}\label{sec:evoinsimulations}

We can extend the measurement operator defined in \ref{eq:cont_sim_povm} to arbitrary operators.  We define
\begin{equation}\label{eq:cont_sim_op}
C(M) =  \proj{0}{0} \otimes M + \proj{1}{1} \otimes M^{*}.
\end{equation}
Note that $C(M)$ can be expressed differently as
\begin{equation}\label{eq:cont_sim_op2}
C(M) = I \otimes Re(M) + iZ \otimes Im(M)
\end{equation}
where $Re(M)$ and $Im(M)$ are the real and imaginary parts of $M$ (both real matrices).  In the case of a multi-party simulation, the $Z$ operates on a particular party's added qubit.

We summarize some of the properties of $C(M)$ here

\begin{lemma}
Let $M$ and $N$ be matrices.  Then we have the following:
\begin{enumerate}
	\item $C(MN) = C(M)C(N)$.
	\item $C(M + N) = C(M) + C(N)$.
	\item Let $a$ be a real number, then $C(aM) = aC(M)$.
	\item If $\ket{\psi}$ is an eigenvector of $M$ with eigenvalue $\lambda$, then $\ket{0}\ket{\psi}$ and $\ket{1}\ket{\psi}$ are eigenvectors of $C(M)$ with eigenvalues $\lambda$ and $\lambda^{*}$, respectively.
	\item $C(M)$ is Hermitian if and only if $M$ is.
	\item $C(M)$ is unitary if and only if $M$ is.
	\item $C(M)$ is positive semi-definite if and only if $M$ is.
	\item When $M$ is Hermitian, $\tr{C(M)} = 2 \tr{M}$.
\end{enumerate}
\end{lemma}

These properties can be derived easily.

\subsubsection{Discrete time evolution}
The properties of $C(\cdot)$ allow us to easily determine how the simulation states in the continuum evolve.  Let $U$ and $\ket{\psi}$ be a reference unitary operation and state and let $\rho^{\prime}$ be as in equation~\ref{eq:simstate}.  By the form of $C(U)$ we have

 \[
C(U)\rho^{\prime}C(U)^{\dagger} = 
a \proj{0}{0} \otimes U\proj{\psi}{\psi}U^{\dagger} + (1-a) \proj{1}{1} \otimes U^{*}\proj{\psi^{*}}{\psi^{*}}U^{T} +\]
\[ c \proj{0}{1}\otimes U\proj{\psi}{\psi^{*}}U^{T} + c^{*} \proj{1}{0} \otimes U^{*}\proj{\psi^{*}}{\psi}U^{\dagger}
 .
\]
But this is the simulation state for $U\ket{\psi}$, and hence $C(U)$ evolves the simulation state $\rho^{\prime}$ to produce a new simulation state corresponding to $U\ket{\psi}$.  Compositions of unitaries will also evolve the state correctly so that the measurement statistics at the end of a circuit will be identical to that of the reference circuit.

General quantum operations may be mapped similarly.  It is easy to verify that in Kraus representation a completely positive map is mapped to a completely positive map if we apply $C(\cdot)$ to each of the Kraus operators.  The trace preserving property is also preserved.  We apply the same reasoning as for $U$ above to with each Kraus operator.  The linearity of $C$ then allows us to conclude that the simulation map will behave correctly.  That is to say, it will map $\rho^{\prime}$ to a new simulation state corresponding to $\ket{\psi}$ evolved under the reference map.

\subsubsection{Continuous time evolution}
We begin with a Hamiltonian $H$.
One can simulate the Schr\"{o}dinger evolution of $H$ on $\ket{\psi}$ by evolving $H^{*}$ on $\ket{\psi^*}$ backwards in time, or equivalently, evolving the system according to $-H^{*}$, and measuring with conjugated observables.

Thus, the simulation of the evolution of $H$ can be achieved using the Hamiltonian
\begin{equation}
H^{\prime} = \proj{0}{0} \otimes H - \proj{1}{1} \otimes H^{*}.
\end{equation}
The evolution of the state according to the Schr\"{o}dinger equation
\begin{equation}
U(t) = e^{-iH^{\prime}t}
\end{equation}
gives
\begin{equation}
e^{-iH^{\prime}t} = \proj{0}{0} \otimes e^{-iHt}  + \proj{1}{1} \otimes e^{-i(-H^{*})t} = \proj{0}{0} \otimes e^{-iHt}  + \proj{1}{1} \otimes \left(e^{-iHt}\right)^{*} = C(e^{-iHt})
\end{equation}
(using the fact that $\exp(A + B) = \exp(A) +  \exp(B)$ when $AB = 0 = BA$, and that  $\exp\left(P \otimes A\right) = P \otimes \exp( A)$ when $P^{2} = P$).
Thus
\begin{equation}
e^{-iH^{\prime}t} = C(e^{-iHt})
\end{equation}
and the simulation evolution tracks that of the reference system.

Another way to arrive at the same $H^{\prime}$ is the approach used in the real simulation \cite{McKague:2009:Simulating-Quan}.  There, rather than considering the Hamiltonian alone, the whole matrix in the exponent, $-iHt$, was considered.  Applying $C(\cdot)$ to this matrix we obtain
\begin{equation}
\proj{0}{0} \otimes (-iHt) + \proj{1}{1} \otimes (-iHt)^{*} =  i\left(\proj{0}{0} \otimes H - \proj{1}{1} \otimes H^{*} \right)t
\end{equation}
Here the fact that $a^{*}b^{*} = (ab)^{*}$ means $(-iH)^{*} = iH^{*}$ and the $-1$ factor is explained.

\section{Real simulation in the family}\label{sec:real_sim_in_cont_sim}

The real simulation presented in \cite{McKague:2009:Simulating-Quan} can be expressed as a simulation in the family defined above through a change of basis.  Starting with the state defined as in \ref{eq:simstate} with $a = c  = \frac{1}{2}$ the simulation state is pure and equal to
 \[
\ket{\psi^{\prime}} = \frac{1}{\sqrt{2}} \ket{0}\ket{\psi} + \frac{1}{\sqrt{2}} \ket{1}\ket{\psi^{*}}.
\]
We next apply a Hadamard gate followed by the relative phase rotation
 \[
\left(
	\begin{matrix}
	1 & 0 \\
	0 & -i \\
	\end{matrix}
\right)
\]
to the extra qubit.  This is the same as applying the unitary
\begin{equation}
U = \left(
	\begin{matrix}
	1  & 1 \\
	-i & i \\
	\end{matrix}
\right).
\end{equation}
The resulting state is
 \[
\frac{1}{2}\ket{0}(\ket{\psi} + \ket{\psi^{*}}) - \frac{i}{2}\ket{1}(\ket{\psi} - \ket{\psi^{*}})
\]
which can be rewritten as
 \[
\ket{0}Re(\ket{\psi}) + \ket{1}Im(\ket{\psi})
\]
which is the real simulation described in \cite{McKague:2009:Simulating-Quan}\footnote{This part of the real simulation was previously well known}.

Operators are transformed quite easily.  For operator $M$ we conjugate $C(M)$ by $U \otimes I$.  From \ref{eq:cont_sim_op2} we see that the resulting operator is
\begin{equation}
(U \otimes I) C(M)(U^{\dagger} \otimes I) = I \otimes Re(M) + XZ \otimes Im(M)
\end{equation}
which is exactly the operator used in the real simulation for $M$.

The states used in the multi-party simulation in \cite{McKague:2009:Simulating-Quan} are stabilized by $Y_{s}\otimes Y_{t}$ for distinct $s,t$.  Also note that the states used in the simulations defined here are stabilized by $Z_{s}\otimes Z_{t}$ for distinct $s,t$.  The qubit-wise transformation applied transformations $Z$ into $Y$, so the multi-party states are also transformed correctly.

\section{Simplified proof for Mayers-Yao self-test}\label{sec:mayers-yaoproof}

\subsection{Proof Overview}
The main advantages of the following new proof for the Mayers-Yao self-test is that it is shorter, clearer, and more naturally extends to the more general test given in this paper.

The proof has two distinct parts.  The first part establishes some equations on the state and observables based on the observed statistics.  These are straightforward and are a direct result of the statistics observed.  Next we use these equations to show that the $X$ and $Z$ observables on each side anti-commute on the support of the state.  The second part uses the anti-commuting observables to construct local isometries that take the state and observables to the reference state and observables.

One important consideration is that of the support of the state.  Since we do not make any claims about the state and observables outside the support of the state we disregard the rest of the Hilbert space.  In this way we will not make any more reference to the support of the state.

\subsection{Observed statistics imply anti-commuting observables}\label{sec:anticommute}

\subsubsection{Statistics}

In the reference test the marginals for each observable are all 0.  That is,
\[
	\bra{\phi_{+}}M \otimes I \ket{\phi_{+}} = 0
\]
for $M \in \{X, Z, D \}$.  (Swapping the systems in this and the following equations gives the same result since $\ket{\phi_{+}}$ is symmetric.)  Measuring the same observable on both sides always give identical outcomes.  Thus
\[
\bra{\phi_{+}}M \otimes M\ket{\phi_{+}} = 1.
\]
Additionally, $X$ and $Z$ measurements are uncorrelated.
 \[
\bra{\phi_{+}} X \otimes Z \ket{\phi_{+}}  = 0.
\]
The interesting part comes when we measure $X$ or $Z$ on one side and $D$ on the other.
\[
\bra{\phi_{+}} X \otimes D \ket{\phi_{+}} = \bra{\phi_{+}} Z \otimes D \ket{\phi_{+}} = \frac{1}{\sqrt{2}}.
\]

\subsubsection{State equalities}

Using the equations from section~\ref{sec:reftest} on the measurement outcomes combined with the fact that $\ket{\psi}$ is normalized gives us the following equations

\begin{eqnarray}
\ket{\psi} & = & X_{A} \otimes X_{B} \ket{\psi} \\
& = & Z_{A} \otimes Z_{B} \ket{\psi} \\
& = & D_{A} \otimes D_{B} \ket{\psi} \\
X_{A} \otimes I \ket{\psi} & = & I \otimes X_{B} \ket{\psi} \label{eq:xAxB} \\
Z_{A} \otimes I \ket{\psi} & = & I \otimes Z_{B} \ket{\psi} \label{eq:zAzB} \\
D_{A} \otimes I \ket{\psi} & = & I \otimes D_{B} \ket{\psi} \\
X_{A}Z_{A} \otimes I \ket{\psi} & = & I \otimes Z_{B}X_{B} \ket{\psi} \\
Z_{A}X_{A} \otimes I \ket{\psi} & = & I \otimes X_{B}Z_{B} \ket{\psi} \\
X_{A}Z_{A} \otimes I \ket{\psi} & = & X_{A} \otimes Z_{B} \ket{\psi} \\
Z_{A}X_{A} \otimes I \ket{\psi} & = & Z_{A} \otimes X_{B} \ket{\psi}
\end{eqnarray}

We can also establish some orthogonality relationships between various vectors.  In particular the vectors $\ket{\psi}, X_{A} \otimes I \ket{\psi}, Z_{A} \otimes I \ket{\psi}, X_{A}Z_{A} \otimes I \ket{\psi}$ are pairwise orthogonal.

Our goal for the remainder of the proof is to show that any state for which these equations hold must be equivalent to $\ket{\phi_{+}}$.

\subsubsection{Anti-commuting observables}\label{sec:my_anticommuting}

We now move to more salient matters.  First, we note that  $D_{A} \otimes I \ket{\psi}$ must be in the space spanned by $X_{A} \otimes I \ket{\psi}$ and $Z_{A} \otimes I \ket{\psi}$ because it has overlap $\frac{1}{\sqrt{2}}$ with each of these orthogonal vectors, and it has norm 1.  Thus
 \[
 D_{A} \otimes I \ket{\psi} = \frac{X_{A} + Z_{A}}{\sqrt{2}} \otimes I \ket{\psi}
\]
and analogously for $I \otimes D_{B} \ket{\psi}$.  This allows us to make the following deductions

\begin{eqnarray*}
\ket{\psi} & = & D_{A} \otimes D_{B} \ket{\psi} \\
& = & \frac{1}{2} (X_{A} + Z_{A}) \otimes (X_{B} + Z_{B}) \ket{\psi} \\
& = & \ket{\psi} + (X_{A}\otimes Z_{B} + Z_{A} \otimes X_{B}) \ket{\psi}
\end{eqnarray*}

\noindent
Applying equations~\ref{eq:xAxB} and~\ref{eq:zAzB} we obtain

\begin{equation}
(X_{A}Z_{A} + Z_{A} X_{A}) \otimes I \ket{\psi} = 0.
\end{equation}
By Lemma~\ref{lemma:anticommute_support}, below,  it follows that $X_{A}$ and $Z_{A}$ anti-commute on the support of $\ket{\psi}$ on $A$.  Similarly, the observables $X_{B}$ and $Z_{B}$ anti-commute on support of $\ket{\psi}$ on $B$.

\begin{lemma}\label{lemma:anticommute_support}
Let $X_{A}$ and $Z_{A}$ be operators and $\ket{\psi}_{AB}$ a bipartite state such that
\begin{equation}
X_{A}Z_{A}\otimes I_{B} \ket{\psi}_{AB} = - Z_{A}X_{A} \otimes I_{B} \ket{\psi}_{AB}.
\end{equation}
then $X_{A}Z_{A} \ket{\phi} = - Z_{A}X_{A} \ket{\phi}$ for any $\ket{\phi}$ in the support of $\ket{\psi}_{AB}$ on $A$.
\end{lemma}

\begin{proof}
Let
\begin{equation}
\ket{\psi} = \sum_{j} \lambda_{j} \ket{j}_{A} \ket{j}_{B}.
\end{equation}
be the singular value decomposition of $\ket{\psi}$.  We then have
\begin{equation}
X_{A}Z_{A}\otimes I_{B} \sum_{j} \lambda_{j} \ket{j}_{A} \ket{j}_{B} = - Z_{A}X_{A} \otimes I_{B} \sum_{j} \lambda_{j} \ket{j}_{A} \ket{j}_{B}.
\end{equation}
We now take the inner product with $\ket{k}_{A}\ket{k^{\prime}}_{B}$ for some $k, k^{\prime}$ to obtain
\begin{equation}
\lambda_{j}\bra{k}_{A}X_{A}Z_{A}\ket{j}_{A} = -\lambda_{j}\bra{k}_{A}X_{A}Z_{A}\ket{j}_{A}
\end{equation}
When we restrict to the subspace to the subspace spanned by the $\ket{k}_{A}$ for which $\lambda_{k} \neq 0$ (i.e. on the support of $\ket{\psi}$ on $A$) we find that $X_{A}Z_{A} = - Z_{A}X_{A}$.

\end{proof}

\subsection{Local unitary transformations}

Now we can easily build the local unitaries required to extract the EPR pair.  We use the circuit shown in figure~\ref{fig:epr_local_unitary_circuit}.  There the outer $\ket{0}$ states are added while the two inner wires carry the two halves of the bipartite state $\ket{\psi}$.  This circuit essentially builds a SWAP gate out of two CNOT gates (the usual third gate is not necessary since we initialize with $\ket{0}$.)  The SWAP gate extracts the entanglement out of $\ket{\psi}$ and swaps in a product state.

\begin{figure}
\[
\Qcircuit @C=0.5cm @R=0.5cm{
\lstick{\ket{0}}  & \gate{H} & \ctrl{1}  & \gate{H} & \ctrl{1} &\qw \\
 &  & \gate{Z_{A}} & \qw & \gate{X_{A}} & \qw\\
 \lstick{\ket{\psi}} \ar@{-}[ur] \ar@{-}[dr] & & & & & \\
              &    & \gate{Z_{B}} & \qw & \gate{X_{B}} & \qw\\
\lstick{\ket{0}}     & \gate{H} & \ctrl{-1} & \gate{H} & \ctrl{-1} & \qw
}
\]
\caption{Circuit for $\Phi$ showing equivalence of physical circuit to reference circuit in Mayers-Yao test}
\label{fig:epr_local_unitary_circuit}
\end{figure}
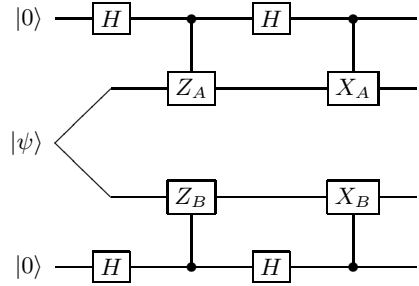

The circuit gives two isometries, one for each wire in EPR test circuit, which we denote $\Phi_{A}$ and $\Phi_{B}$.

\subsubsection{State}

After applying this circuit the resulting state is
 \begin{eqnarray*}
 \Phi_{A} \otimes \Phi_{B} ( \ket{\psi}) & = &\frac{1}{4} (I + Z_{A}) \otimes (I + Z_{B}) \ket{\psi}\ket{00} \\
 & + & \frac{1}{4}(I + Z_{A}) \otimes X_{B}(I - Z_{B}) \ket{\psi}\ket{01} \\
 & + & \frac{1}{4}X_{A}(I- Z_{A}) \otimes (I + Z_{B}) \ket{\psi}\ket{10}\\
 & + & \frac{1}{4} X_{A} (I - Z_{A}) \otimes X_{B} (I-Z_{B}) \ket{\psi}\ket{11} \\
\end{eqnarray*}
Applying some equations and the anti-commuting result from the previous section we find that this is equal to
 \[
  \Phi_{A} \otimes \Phi_{B} ( \ket{\psi}) = \frac{1}{4}(I+Z_{A}) \otimes (I + Z_{B}) \ket{\psi}\left(\ket{00} + \ket{11} \right) +
 \]
 \[
 (I + Z_{A})(I - Z_{A})\otimes X_{B}\ket{\psi}\ket{01} +
  X_{A}\otimes (I + Z_{B})(I - Z_{B})\ket{\psi}\ket{10}
\]
 \[
  =  \frac{1}{\sqrt{2}} (I \otimes I + I \otimes Z_{B}) \ket{\psi} \ket{\phi_{+}}\]

This may look curious since $I + Z_{A}$ and $I + Z_{B}$ are not unitary.  In fact it is easy to show that the final state still has the correct norm.  To give some intuition, note that in the reference case we want to extract $\ket{\phi_{+}}$ and swap in $\ket{00} = \frac{1}{2\sqrt{2}}(I+Z) \otimes (I+Z)\ket{\phi_{+}}$.

\subsubsection{Measurement operators}

We now turn to equivalence of the measurement operators.  We start with $X_{A}$ (the result for $X_{B}$ follows analogously).  Applying $X_{A}$ to $\ket{\psi}$ before applying the circuit is the same as applying it at the end, with a $-1$ phase introduced by anti-commuting past the controlled $Z_{A}$ operation (recall from section~\ref{sec:my_anticommuting} that $X_{A}$ and $Z_{A}$ anti-commute on the relevant subspace).  The resulting state is
 \begin{eqnarray*}
  \Phi_{A} \otimes \Phi_{B} ( X_{A} \otimes I_{B}\ket{\psi})& =  &\frac{1}{4} X_{A}(I - Z_{A}) \otimes (I + Z_{B}) \ket{\psi}\ket{00} \\
 & + & \frac{1}{4}X_{A}(I - Z_{A}) \otimes X_{B}(I - Z_{B}) \ket{\psi}\ket{01} \\
 & + & \frac{1}{4}(I+ Z_{A}) \otimes (I + Z_{B}) \ket{\psi}\ket{10}\\
 & + & \frac{1}{4}  (I + Z_{A}) \otimes X_{B} (I-Z_{B}) \ket{\psi}\ket{11} \\
\end{eqnarray*}
Following the same logic as used in the state equivalence, we find that the final state is
 \[
 \Phi_{A} \otimes \Phi_{B} ( X_{A} \otimes I \ket{\psi})  = \frac{1}{\sqrt{2}} (I \otimes I + I \otimes Z_{B}) \ket{\psi}(X \otimes I) \ket{\phi_{+}}\]

For the $Z_{A}$ operation, we see that the effect is a $-1$ phase kicked back through the final controlled $X_{A}$ operation.  This phase appears on the terms with $\ket{1}$ in the qubit, exactly as if a $Z$ operation had been applied to the qubit.  The equivalence for the $D$ operators results from the fact that $D = \frac{X + Z}{\sqrt{2}}$ on the relevant subspace, and linearity.

This concludes the proof of Theorem~\ref{theorem:mayers-yao} .

\end{document}